%% file: main.tex
\documentclass[11pt]{article}
\input{preamble}

\usepackage{fullpage}
\begin{document}

\title{Improved Algorithms for White-Box Adversarial Streams}
\author{Ying Feng\thanks{Carnegie Mellon University. 
E-mail: \email{yingfeng@andrew.cmu.edu}}
\and
David P. Woodruff\thanks{Carnegie Mellon University. 
E-mail: \email{dwoodruf@andrew.cmu.edu}}
}
\date{\today}

\maketitle

\begin{abstract}
We study streaming algorithms in the white-box adversarial stream model, where the internal state of the streaming algorithm is revealed to an adversary who adaptively generates
the stream updates, but the algorithm obtains fresh randomness unknown to the adversary at each time step. We incorporate cryptographic assumptions to construct robust algorithms against such adversaries. We propose efficient algorithms for sparse recovery of vectors, low rank recovery of matrices and tensors, as well as low rank plus sparse recovery of matrices, i.e., robust PCA. Unlike deterministic algorithms, our algorithms can report when the input is not sparse or low rank even in the presence of such an adversary. We use these recovery algorithms to improve upon and solve new problems in numerical linear algebra and combinatorial optimization on white-box adversarial streams. For example, we give the first efficient algorithm for outputting a matching in a graph with insertions and deletions to its edges provided the matching size is small, and otherwise we declare the matching size is large. We also improve the approximation versus memory tradeoff of previous work for estimating the number of non-zero elements in a vector and computing the matrix rank.
\end{abstract}

\newpage

\section{Introduction}
The streaming model captures key resource requirements of algorithms for database, machine learning, and network tasks, where the size of the data is significantly larger than the available storage, such as for internet and network traffic, financial transactions, simulation data, and so on. This model was formalized in the work of Alon, Matias, and Szegedy \cite{10.1145/237814.237823}, which models a vector undergoing additive updates to its coordinates. Formally, there is an underlying $n$-dimensional vector $x$, which could be a flattened matrix or tensor, which is initialized to $0^n$ and evolves via an arbitrary sequence of $m \leq \poly(n)$ additive updates to its coordinates. These updates are fed into a streaming algorithm, and the $t$-th update has the form $(i_t, \delta_t)$, meaning $x_{i_t} \leftarrow x_{i_t} + \delta_t$. Here $i_t \in \{1, 2, \ldots, n\}$ and $\delta_t \in \{-M, -M+1,\ldots , M-1, M\}$ for an $M \leq \poly(n)$\footnote{All bounds can be generalized to larger $m$ and $M$; this is only for convenience of notation. Also, our streaming model, which allows for both positive and negative updates, is referred to as the (standard) turnstile streaming model in the literature.}. Throughout the stream, $x$ is promised to be in $\{-M, -M+1,\ldots , M-1, M\}^n$.
A streaming algorithm makes one pass over the stream of updates and uses  limited memory to approximate a function of $x$.

A large body of work on streaming algorithms has been designed for {\it oblivious} streams, for which the sequence of updates may be chosen adversarially, but it is chosen independently of the randomness of the streaming algorithm. In practical scenarios, this assumption may not be reasonable; indeed, even if the stream is not generated by an adversary this may be problematic. For example, if one is running an optimization procedure, then one may feed future data into a streaming algorithm based on past outputs of that algorithm, at which point the inputs depend on the algorithm's randomness and there is no guarantee of correctness. This is also true in recommendation systems, where a user may choose to remove suggestions based on previous queries. 

There is a growing body of work on streaming algorithms that are robust in the black-box adversarial streaming model \cite{Ben-EliezerY20,Ben-EliezerJWY21,HassidimKMMS20,WoodruffZ21,AlonBDMNY21,KaplanMNS21,BravermanHMSSZ21,MenuhinN21,AttiasCSS21,Ben-EliezerEO22,ChakrabartiGS22}, in which the adversary can monitor only the output of the streaming algorithm and choose future stream updates based on these outputs. While useful in a number of applications, there are other settings where the adversary may also have access to the internal state of the algorithm, and this necessitates looking at a stronger adversarial model known as white-box adversarial streams. 

\subsection{The White-box Adversarial Streaming Model}

We consider the white-box adversarial streaming model, introduced in \cite{10.1145/3517804.3526228}, where a sequence of stream updates $u_1,\ldots,u_m$ is chosen adaptively by an adversary who sees the full internal state of the algorithm at all times, including the parameters and the previous randomness used by the algorithm.

\begin{defn}[White-box Adversarial Streaming Model] Consider a single-pass, two-player game between \textsf{Streamalg}, the streaming algorithm, and \textsf{Adversary}. Prior to the beginning of the game, fix a query $\mathcal{Q}$, which asks for a function of an underlying dataset that will be implicitly defined by the stream, which is itself chosen by \textsf{Adversary}.
    The game then proceeds across $m$ rounds, where in the $t$-th round:
    
    \begin{enumerate}[label=\arabic*.]
    \normalfont{\item \textsf{Adversary} computes an update $u_t$ for the stream, which depends on all previous stream updates, all previous internal states, and randomness used by \textsf{Streamalg} (and thus also, all previous outputs of \textsf{Streamalg}).  }
    
    \normalfont{\item \textsf{Streamalg} acquires a fresh batch $R_t$ of random bits, uses $u_t$ and $R_t$ to update its data structures $D_t$, and (if asked) outputs a response $A_t$ to the query $\mathcal{Q}$.}
    
    \normalfont{\item \textsf{Adversary} observes the response $A_t$, the internal state $D_t$ of \textsf{Streamalg}, and the random bits $R_t$.}
    \end{enumerate}
    
The goal of \textsf{Adversary} is to make \textsf{Streamalg} output an incorrect response $A_t$ to the query $\mathcal{Q}$ at some time $t\in[m]$ throughout the stream.
\end{defn}

{\bf Notation:}
    A function $f(n)$ is said to be {\it negligible} if for every polynomial $P(n)$, for all large enough $n$, $f(n) < \frac{1}{P(n)}$. We typically denote negligible functions by negl$(n)$.

\vspace{0.2em}
Given a fixed time bound $\mathcal{T}$, we say a streaming algorithm is robust against $\mathcal{T}$ time-bounded white-box adversaries if no $\mathcal{T}$ time-bounded white-box adversary can win the game with non-negligible probability against this algorithm.

\subsubsection{Applications of white-box adversaries}

The white-box adversarial model captures characteristics of many real-world attacks, where an adaptive adversary has access to the entirety of the internal states of the system. In comparison to the oblivious stream model or the black-box adversarial model \cite{Ben-EliezerJWY21}, this model allows us to model much richer adversarial scenarios.

For example, consider a distributed streaming setting where a centralized server collects statistics of a database generated by remote users. The server may send components of its internal state $S$ to the remote users in order to optimize the total communication over the network. The remote users may use $S$ in some process that generates downstream data. Thus, future inputs depend on the internal state $S$ of the streaming algorithm run by the central coordinator. In such settings, the white-box robustness of the algorithms is crucial for optimal selection of query plans~\cite{SelingerACLP79}, online analytical processing~\cite{ShuklaDNR96, PadmanabhanBMCH03}, data integration~\cite{BrownHMPRS05}, and data warehousing~\cite{DasuJMS02}.

Many persistent data structures provide the ability to quickly access previous versions of information stored in a repository shared across multiple collaborators. The internal persistent data structures used to provide version control may be accessible and thus visible to all users of the repository. These users may then update the persistent data structure in a manner that is not independent of previous states~\cite{DriscollSST89,FiatK03,Kaplan04}.

Dynamic algorithms often consider an adaptive adversary who generates the updates upon seeing the entire data structure maintained by the algorithm during the execution~\cite{Chan10,ChanH21,RoghaniSW22}. For example, \cite{Wajc20} assumes the entire state of the algorithm (including the set of randomness) is available to the adversary after each update, i.e., a white-box model.

Moreover, robust algorithms and adversarial attacks are important topics in machine learning~\cite{SzegedyZSBEGF14,GoodfellowSS14}, with a large body of recent literature focusing on adversarial robustness of machine learning models against white-box attacks~\cite{IlyasEM18,MadryMSTV18,SchmidtSTTM18,TramerKPGBM18,CubukZMVL18,KurakinGB17,LiuCLS17}. There exist successful attacks that use knowledge of the trained model; e.g., the weights of a linear classifier
to minimize the loss function (which are referred to as Perfect Knowledge adversaries in \cite{BiggioCMNSLGR13}). There are also white-box attacks that use the architecture and parameters of a trained neural network policy to generate adversarial perturbations that are almost imperceptible to the human eye but result in misclassification by the network \cite{HuangPGDA17}. In comparison to the black-box adversarial streaming model, in which the input is chosen by an adversary who repeatedly queries for only a fixed property of the underlying dataset at each time but does not see the full internal state of the algorithm during execution, the white-box model more effectively captures the full capability of these attacks.

\subsection{Random Oracle Model}

In order to construct streaming algorithms based on the hardness of the SIS problem, we need access to a fixed uniformly random matrix during the stream. 
In this paper, we consider algorithms in the {\it random oracle model}, which means that the algorithms, as well as the white-box adversaries, are given read access to an arbitrarily long string of random bits. Each query gives a uniformly random value from some output domain and repeated queries give consistent answers. The random oracle model is a well-studied model and has been used to design
numerous cryptosystems \cite{10.1145/168588.168596, 10.1007/3-540-68339-9_34, 10.1145/1008731.1008734, cryptoeprint:2015/140}. Also, such a model has been used to design space-efficient streaming algorithms, for both oblivious streams \cite{Clifford2013ASS, 10.1145/3596494} as well as adversarial settings \cite{10.1145/3517804.3526228, 10.1145/3375395.3387658}. In the random oracle model, instead of storing large random sketching matrices during the stream, the streaming algorithms can
generate the columns of the matrix on the fly when processing updates. Also, in the distributed setting, the servers will be able to agree on a random sketching matrix without having to communicate it. 

Such an oracle is often implemented with hash-based heuristic functions such as AES or SHA256. These implementations are appealing since they behave, as far as we can tell in practice, like random functions. They are also extremely fast and incur no memory cost.

Another approach is to use a {\it pseudorandom function} as a surrogate.

\begin{defn}[Pseudorandom Function]
\label{def: prf}
Let $A, B$ be finite sets, and let $\mathcal{F} = \{F_i: A \rightarrow B\}$ be a function family, endowed with an efficiently sampleable distribution. We say that $\mathcal{F}$ is a pseudorandom function (PRF) family if all functions $F_i$ are efficiently computable and the following two games are computationally indistinguishable:

    \begin{enumerate}[label=(\roman*), ref=(\roman*)]
        \item\label{case1} Sample a function $F \stackrel{\$}\gets \mathcal{F}$ and give the adversary adaptive oracle access to $F(\cdot)$.
        
        \item\label{case2} Choose a uniformly random function $U: A \rightarrow B$ and give the adversary adaptive oracle access to $U(\cdot)$.
    \end{enumerate}
\end{defn}

Given a random key to draw $F$ from $\mathcal{F}$, a pseudorandom function provides direct access to a deterministic sequence of pseudorandom bits. This pseudorandom bit sequence can be seen as indexed by indices in $A$. Moreover, the key size can be logarithmically small with respect to the function domain (though in our algorithms we only need the key size to be polynomially small).

In this work, we design algorithms based on hardness assumptions of lattice cryptographic problems. In particular, we use the Short Integer Solution (SIS) Problem; see Section \ref{sec:prelim} for the precise cryptographic assumptions we make. There are many existing schemes to construct families of pseudorandom functions based on cryptographic assumptions \cite{10.1145/6490.6503, prfFromLattice, PRFthesis}. Therefore, if we assume the hardness of the SIS problem against $\mathcal{T}$ time-bounded adversaries, then we can construct families of pseudorandom functions. Moreover, if a function $F \stackrel{\$}\gets \mathcal{F}$ is sampled privately from any of these families $\mathcal{F}$, then $F$ behaves just like a random oracle from the perspective of any $\mathcal{T}$ time-bounded adversary. 

However, in the white-box adversarial setting, the process of choosing $F \stackrel{\$}\gets \mathcal{F}$ is revealed to the adversary. So the adversary can distinguish cases \ref{case1} and \ref{case2} in Definition \ref{def: prf} by simply comparing the output with $F$. It may be possible to use a pseudorandom function in place of a random oracle in our algorithms if one can resolve the following question:

\vspace{0.5em}
{\it
Let $\mathcal{F}$ be a family of pseudorandom functions, constructed based on the SIS problem. Consider a one-round, two player game between \textsf{Challenger} and \textsf{Adversary}:}

\begin{enumerate}[label=\arabic*.]
    \item \textsf{Challenger} samples a pseudorandom function $F \stackrel{\$}\gets \mathcal{F}$ based on some random key $\mathcal{K}$, and reveals $\mathcal{K}$ to \textsf{Adversary}.

    \item \textsf{Challenger} uses the pseudorandom bits generated by $F$ to sample an instance $\mathcal{I}$ of the SIS problem, with hardness parameter $n = \lvert \mathcal{K} \rvert$.

    \item \textsf{Adversary} attempts to solve $\mathcal{I}$.
\end{enumerate}

{\it Assuming that no $\mathcal{T}$ time-bounded adversary can solve the SIS problem with non-negligible probability, does there exist a $\mathcal{T}$ time-bounded adversary that can win this game with non-negligible probability, for a fixed time bound $\mathcal{T}$?
}

\vspace{0.5em}
The answer to this question depends on the specific PRF construction that we use. One may be able to artificially construct a family of SIS-based PRFs and show that the pseudorandomness generated by such PRF induces an easy variant of the SIS problem. However, using other PRF constructions, the SIS problem could potentially retain its difficulty. We leave the question of removing our random oracle assumption as an interesting direction for future work. 

\subsection{Our Contributions}

 Table \ref{tab:sum} summarizes our contributions. Specifically, we construct sparse recovery schemes for vectors, low rank plus sparse recovery schemes for matrices, and low rank recovery schemes for tensors, and apply these as building blocks to solve a number of problems in the white-box adversarial streaming model. Our algorithms either improve the bounds of existing algorithms, often optimally, or solve a problem for which previously no known white-box adversarial streaming algorithm was known.

\input{table}

\subsubsection{Recovery Algorithms}

We start by giving recovery algorithms for $k$-sparse vectors and rank-$k$ matrices in the white-box adversarial streaming model, which reconstruct their input provided that it satisfies the sparsity or rank constraint. Our algorithms crucially have the property that {\it they can detect if their input violates the sparsity or rank constraint}. 

Our algorithms make use of hardness assumptions of the Short Integer Solution (SIS) Problem and hold against polynomial (and sometimes larger) time adversaries. See Section \ref{sec:prelim} for the precise cryptographic assumptions we make. 
Informally, we have
Theorem \ref{thm:k-rec-1}, Theorem \ref{thm:m-rec-1}, Theorem 
\ref{thm:rpca-1}, and Theorem
\ref{thm:tensor} below. 

\begin{theorem}
\label{thm:k-rec-1}
    Assuming the exponential hardness of the SIS problem, there exists a white-box adversarially robust streaming algorithm which determines if the input vector is $k$-sparse, for parameter $k \geq n^c$ for an arbitrarily small constant $c > 0$, and if so, recovers a $k$-sparse vector using $\tilde{\mathcal{O}}({k})$ bits\footnote{Here and throughout, $\tilde{\mathcal{O}}(f)$ denotes $f \cdot \poly(\log n)$, with $n$ defined in our description of the streaming model.} of space in the random oracle model.
\end{theorem}

We note that there are standard deterministic, and hence also white-box adversarially robust, $k$-sparse vector recovery schemes based on any deterministic algorithm for compressed sensing \cite{CW08}.  However, previous algorithms require the promise that the input is $k$-sparse; otherwise, their output can be arbitrary. That is, there is no way to know if the input is $k$-sparse or not. In contrast, our algorithm does not assume sparsity of the input, and reports a failure when the input is not $k$-sparse. We stress that this is not an artifact of analyses of previous algorithms; in fact, any deterministic streaming algorithm cannot detect if its input vector is $k$-sparse without using $\Omega(n)$ bits of memory \cite{GM06}. By Theorem 2 in ~\cite{10.1145/3517804.3526228}, this implies an $\Omega(n)$ bit lower bound for {\it any randomized} $k$-sparse decision algorithms in the white-box streaming model. Thus our algorithm provides a provable separation between computationally bounded and unbounded adversaries, under cryptographic assumptions. 

While sparsity is a common way of capturing vectors described with few parameters, low rank is a common way of capturing matrices described with few parameters. We next extend our results to the matrix setting:

\begin{theorem}
\label{thm:m-rec-1}
    Assuming the exponential hardness of the SIS problem, there exists a white-box adversarially robust streaming algorithm which decides if an $n \times n$ input matrix with integer entries bounded by a polynomial in $n$, has rank at most $k$, and if so, recovers the matrix using $\tilde{O}(nk)$ bits of space in the random oracle model \footnote{Our results all generalize to $n \times d$ matrices; we state them here for square matrices for convenience only.}.
\end{theorem}

Theorem \ref{thm:m-rec-1} provides the first low rank matrix recovery algorithm in the white-box streaming model. Moreover, the space complexity of this algorithm is nearly optimal, as just describing such a matrix requires $\Omega(nk \log n)$ bits. This result again provides a separation from deterministic algorithms under cryptographic assumptions, as a simple reduction from the Equality communication problem (see, e.g., \cite{AlonMS99} for similar reductions) shows that testing if the input matrix in a stream is all zeros or not requires $\Omega(n^2)$ memory. 

In addition, our results can be further extended to recover a sparse plus low rank matrix for robust principal component analysis (robust PCA) \cite{doi:10.1137/090761793, 10.1145/1970392.1970395} and also can recover a low rank tensor: 

\begin{theorem}[Informal]
    Assuming the exponential hardness of the SIS problem, for $k \geq n^{\Omega(1)}$, there exists a white-box adversarially robust streaming algorithm which decides if an $n \times n$ input matrix can be decomposed into a matrix with rank at most $k$ plus a matrix with at most $r$ non-zero entries using $\tilde{\mathcal{O}}(nk+r)$ bits of space in the random oracle model.
\end{theorem}

\begin{theorem}[Informal]
    Assuming the exponential hardness of the SIS problem, there exists a white-box adversarially robust streaming algorithm which decides if an $n_1 \times \cdots \times n_d$ input tensor has CP\mbox{-}rank at most $k$ for some $k \geq (\prod_{i=1}^d n_i)^{\Omega(1)}$, and if so, recovers the tensor using $\tilde{\mathcal{O}}(k(n_1 + \cdots n_d))$ bits of space in the random oracle model.
\end{theorem}

\subsubsection{Applications}
Our sparse recovery theorem for vectors can be used to simplify and improve the existing upper bound for the $\ell_0$-norm estimation problem in the white-box adversarial model, which is the problem of estimating the support size, i.e., the number of non-zero entries of the input vector $x$. 

\begin{theorem}[Informal]
    Assuming the exponential hardness of the SIS problem, there exists a white-box adversarially robust streaming algorithm which estimates the $\ell_0$ norm within a factor of $n^\eps$ using $\tilde{\mathcal{O}}(n^{1-\eps})$ bits of space in the random oracle model.
\end{theorem}

Previously, the only known white-box adversarily robust algorithm for $\ell_0$ norm estimation required $\Tilde{\mathcal{O}}(n^{1-\eps+c\eps})$ space for an $n^\eps$-approximation, where $c > 0$ is a fixed constant, in the random oracle model. Our algorithm replaces $c$ with $0$.

Based on our low rank matrix recovery algorithm, we give the first algorithm for finding a maximum matching in a graph if the maximum matching size is small, or declare that the maximum matching size is large, in a stream with insertions and deletions to its edges. Standard methods based on filling in the so-called Tutte matrix of a graph randomly do not immediately work, since the adversary sees this randomness in the white box model. Nevertheless, we show that filling in the Tutte matrix deterministically during the stream suffices for our purposes. 

\begin{theorem}[Informal]
\label{thm:apps-1}
    Assuming the exponential hardness of the SIS problem, there is a white-box adversarially robust streaming algorithm using $\tilde{O}(nk)$ space in the random oracle model and $\poly(n)$ running time, which either declares the maximum matching size is larger than $k$, or outputs a maximum matching.  
\end{theorem}
We note that for any matrix problem, such as linear matroid intersection or parity or union, matrix multiplication and decomposition, finding a basis of the null space, and so on, if the input consists of low rank matrices then we can first recover the low rank matrix in a white-box adversarial stream, verify the input is indeed of low rank, and then run an offline algorithm for the problem, such as those in \cite{cheung13}. 

Besides solving new problems, as an immediate corollary we also obtain an improved quantitative bound for testing if the rank of an input matrix is at most $k$, which is the rank decision problem of \cite{10.1145/3517804.3526228}.

\begin{theorem}[Informal] 
Assuming the exponential hardness of the SIS problem, there exists a white-box adversarially robust streaming algorithm which solves the rank decision problem using $\tilde{\mathcal{O}}({nk})$ bits of space in the random oracle model.
\end{theorem}

In \cite{10.1145/3517804.3526228}, a weaker $\mathcal{O}(nk^2)$ space bound was shown for white-box adversarially robust algorithms. Our improvement comes by observing that we can get by with many fewer than $k$ rows in our sketch, provided that the modulus $q$ in the SIS problem (see Section \ref{sec:prelim}) is large enough. This may be counterintuitive as the rank of our sketch may be much less than $k$ but we can still recover rank-$k$ inputs by using a large enough modulus to encode them. 

\section{Preliminaries}\label{sec:prelim}
\subsection{Short Integer Solution Problem}

We make use of well-studied cryptographic assumptions in the design of our algorithms. Specifically, we construct white-box adversarially robust algorithms based on the assumed hardness of the Short Integer Solution problem.

\begin{defn}[Short Integer Solution (SIS) Problem] Let $n, m, q$ be integers and let $\beta > 0$. Given a uniformly random matrix $A \in \mathbb{Z}^{n \times m}_{q}$ with $m \in \poly(n)$, the SIS problem is to find a non-zero integer vector $z \in \mathbb{Z}^m$ such that $Az = 0\mod q$ and $\lVert z\rVert_2 \leq \beta$.
\end{defn}

\begin{theorem}[\cite{cryptoeprint:2013/069}]
\label{thm:SIS}
    Let $n$ and $m, \beta, q \in \poly(n)$ be integers and $q \geq n\cdot \beta$. Then solving the SIS problem with non-negligible probability, with parameters $n, m, q, \beta$ is at least as hard as $\gamma$-approximation of the Shortest Vector Problem (SVP$_\gamma$) with $\gamma \in \poly(n)$.
\end{theorem}

Theorem \ref{thm:SIS} bases the hardness of the SIS problem on the SVP$_\gamma$ problem, which is one of the most well-studied lattice problems with many proposed algorithms. The best known algorithm for SVP$_\gamma$ with $\gamma = \poly(n)$ is due to \cite{10.1145/2746539.2746606} and runs in $\tilde{\mathcal{O}}(2^n)$ time.

\subsection{SIS Hardness Assumption}
\label{sec:sis}

We assume that the white-box adversary is computationally bounded in such a way that it cannot solve the SIS problem with non-negligible probability. For the purposes of this paper, we consider a time bound based on the state-of-the-art complexity result for lattice problems.

\begin{asm} 
\label{asm:strong}
    Given $n \in \mathbb{N}$, for some $m, \beta, q \in \poly(n)$ and $q \geq n\cdot \beta$, no $o(2^{n})$ time-bounded adversary can solve the SIS problem $\mathsf{SIS}_{n, m, p, \beta}$ with non-negligible probability.
\end{asm}

As shown above, 
our instance of the SIS problem is at least as hard as the approximation problem SVP$_{\poly(n)}$, for which the best-known algorithm runs in $\tilde{\mathcal{O}}(2^n)$ time. 

We have the following two crucial lemmas.  

\begin{lem} 
\label{lem:vec}
    Under Assumption \ref{asm:strong}, given a uniformly random matrix $A \in \mathbb{Z}^{n \times m}_q$ for $q, m, \beta \in \poly(n)$ and $q \geq n\cdot \beta$, if a vector $x \in \mathbb{Z}^m_\beta$ is generated by an $o(2^{n})$-time adversary, then with probability at least $ 1-\textrm{negl}(n)$, there does not exist a $k$-sparse vector $y \in \mathbb{Z}^m_\beta$ for which $x \neq y\mod q$ yet $Ax = Ay \mod q$, for $k \in o(\frac{n}{\log n})$.
\end{lem}

\begin{rem}
    We note that given a random matrix $A \in \mathbb{Z}^{n \times m}_q$, when both $x$ and $y$ are $k$-sparse, we can argue information-theoretically by a union bound that with high probability all sparse $x \neq y$ with bounded entries satisfy $Ax \neq Ay$. However, there may exist a binary vector $x$ which is not $k$-sparse, and a $k$-sparse $y$ with bounded entries such that $Ax = Ay$. In this case we need the SIS assumption to show that it is hard for an adversary to find such $x$ and fool the algorithm.
\end{rem}

\begin{proof}
    If an adversary were to find a vector $y \in \mathbb{Z}^m_\beta$ for which $x \neq y\mod q$ yet $Ax = Ay\mod q$, then it would be able to solve the SIS problem by outputting $(x-y)\mod q$, which is a short (i.e., polynomially bounded integer entry), non-zero vector in the kernel of $A$. Because the entries of $y$ are bounded by $q$, it takes at most $\mathcal{O}(q^k\cdot{m \choose k}) \leq \poly(n)^k$ time for an adversary to try all $k$-sparse vectors $y \in \mathbb{Z}^m_\beta$. So it must be that for $k \in o(\frac{n}{\log n})$, such a $k$-sparse vector $y$ does not exist with probability greater than $\textrm{negl}(n)$, as otherwise an $o(2^{n})$-time adversary would be able to find it by enumerating all candidates and use it to solve the SIS problem with non-negligible probability. 
\end{proof}

We similarly have the following lemma:

\begin{lem} 
\label{lem:mat}
    Under Assumption \ref{asm:strong}, given a uniformly random matrix $A \in \mathbb{Z}^{n \times m}_q$ for $q, m, \beta \in \poly(n)$ and $q \geq n\cdot \beta$, if a matrix $X \in \mathbb{Z}^{\sqrt{m} \times \sqrt{m}}_\beta$ is generated by an $o(2^{n})$-time adversary, then with probability at least $ 1-negl(n)$, there does not exist a matrix $Y \in \mathbb{Z}^{\sqrt{m} \times \sqrt{m}}_\beta$ with $rank(Y) \leq k$, such that $X \neq Y\mod q$ and $Ax = Ay\mod q$, for $x, y$ being the vectorizations of $X$ and $Y$, respectively, and $k \in o(\frac{n}{\sqrt{m}\log n})$.
\end{lem}

\begin{proof} As in the proof of Lemma \ref{lem:vec}, an adversary is able to try all matrices $Y \in \mathbb{Z}^{\sqrt{m} \times \sqrt{m}}_\beta$ with $rank(Y) \leq k$ in $\poly(n)^{\sqrt{m}k}$ time. This is because there are $\mathcal{O}{\sqrt{m} \choose k}$ ways of positioning the linearly independent columns of $Y$, with $\poly(n)^{\sqrt{m}k}$ choices for their values when $\beta \in \poly(n)$. All remaining columns are linear combinations of the independent columns. Since there are $\poly(n)^k$-many possible combinations of coefficients and we choose $(\sqrt{m}$-$k)$ of them, there are $\poly(n)^{(\sqrt{m}-k)k}$ choices for the dependent columns. Therefore in total we have $\poly(n)^{\sqrt{m}k}$-many candidate matrices. For $k \in o(\frac{n}{\sqrt{m}\log n})$, there exists an $o(2^{n})$-time adversary that is able to iterate through all candidate matrices. Thus,  under Assumption \ref{asm:strong}, with overwhelming probability such a $Y$ does not exist; otherwise, given $X$ and $Y$, an adversary can easily solve the SIS problem by outputting $(x$-$y)\mod q$.
\end{proof}

\section{Vector Recovery}

\subsection{$k$-Sparse Recovery Algorithm }

\begin{theorem}
    Under Assumption \ref{asm:strong}, given a parameter $k \in \Theta(\frac{n^{c}}{\log n})$ for an arbitrary constant $c > 0$, and a length-$n$ input vector with integer entries bounded by  $\poly(n)$, there exists a streaming algorithm robust against $o(n^k)$ time-bounded white-box adversaries that determines if the input is $k$-sparse, and if so, recovers a $k$-sparse vector using $\tilde{\mathcal{O}}(k)$ bits of space in the random oracle model.
\end{theorem}

{\bf Notation:}
    A function $f(k)$ is said to be in $\omega(k)$ if for all real constants $c > 0$, there exists a constant $k_0 > 0$ such that $f(k) > c\cdot k$ for every $k \geq k_0$.

\input{algorithms/alg1}

\begin{proof}
    Algorithm \ref{alg:k-rec} decides and recovers a $k$-sparse vector using $\tilde{\mathcal{O}}(k)$ bits. The algorithm receives a stream of integer updates to an underlying vector, whose entries are assumed to be at most $\beta \in poly(n)$ at any time. Thus we can interpret the stream updates to be mod $q$ for $q, \beta \in \poly(n)$ and $q \geq n \cdot \beta$.
    
    When an input vector $x$ is $k$-sparse, for a uniformly random sketching matrix $A$, it is guaranteed by Lemma \ref{lem:vec} that $Ay = Ax\mod q$ implies $y=x$. Therefore, in the $k$-sparse case, by enumerating over all $k$-sparse vectors, Algorithm \ref{alg:k-rec} correctly recovers the input vector $y = x$. On the other hand, for inputs that have sparsity larger than $k$, Lemma \ref{lem:vec} guarantees that during post-processing, the enumeration over $k$-sparse vectors will not find a vector $y$ satisfying $Ay=v\mod q$. Thus, in this case Algorithm \ref{alg:k-rec} outputs $None$ as desired. In the random oracle model, we can generate the columns of a uniformly random matrix $A$ on the fly. Then, Algorithm \ref{alg:k-rec} only stores a vector of length $f(k)\cdot \log n$ with entries bounded by $\poly(n)$, so $\tilde{\mathcal{O}}(k)$ bits of space.
\end{proof}

\begin{rem} 
\label{rem:vec}
With roughly $k$ space, any white-box adversarially robust algorithm for $k$-sparse recovery has to assume that the adversary is at most $n^k$-time bounded. Otherwise, given that an algorithm using $k$ words of memory has at most $n^k$ states, for a $k'$-sparse input $x$ with $k'$ slightly larger than $k$, there exists an $x' \neq x$ that goes to the same state as $x$ with high probability. Hence, the adversary would have enough time to find $x$ and $x'$. If the adversary inserts either $x$ or $x'$ in the stream, followed by $-x$, the algorithm cannot tell if the input is $0$ or $x'-x$. Thus, our algorithm is nearly optimal in the sense that it uses $\tilde{\mathcal{O}}(k)$ bits assuming the adversary is $o(2^{k\log n}) = o(n^k)$-time bounded.
\end{rem}

\subsection{Fast $k$-Sparse Recovery}

Algorithm \ref{alg:k-rec} enumerates over all possible $k$-sparse vectors in  post-processing, which is time-inefficient. We now give a faster version of $k$-sparse recovery, which is also capable of identifying whether the input is $k$-sparse. In parallel we run an existing deterministic $k$-sparse recovery scheme that has fast update time assuming the input is $k$-sparse.

\begin{theorem}[\cite{DeterministicCompressedSensing}]
\label{thm:fast}
    There exists a deterministic algorithm that recovers a $k$-sparse length-$n$ vector in a stream using $\tilde{\mathcal{O}}(k)$ bits of space and $\poly(n)$ time.
\end{theorem}

When the input vector is $k$-sparse, the algorithm in Theorem \ref{thm:fast} outputs the input vector. However, when taking in an input vector with sparsity larger than $k$, this algorithm erroneously assumes the input to be $k$-sparse and has no guarantees, in which case the user cannot tell if the output is a correct recovery or not. 
To fix this, we run the two recovery schemes from Algorithm \ref{alg:k-rec} and Theorem \ref{thm:fast} in parallel. Any deterministic recovery scheme is robust against white-box adversaries, and therefore using the algorithm of Theorem \ref{thm:fast} as a subroutine does not break our robustness. 

\begin{theorem}
    Under Assumption \ref{asm:strong}, given a parameter $k \in \Theta(\frac{n^{c}}{\log n})$ for an arbitrary constant $c > 0$, and a length-$n$ input vector with integer entries bounded by  $\poly(n)$, there exists a streaming algorithm robust against $o(n^k)$ time-bounded white-box adversaries that determines if the input is $k$-sparse, and if so, recovers a $k$-sparse vector using $\tilde{\mathcal{O}}(k)$ bits of space and $\poly(n)$ time in the random oracle model.
\end{theorem}

\input{algorithms/alg2}

\begin{proof}
    
    Algorithm \ref{alg:fast} gives a fast version of $k$-sparse recovery. By running the two schemes in parallel, at the end of the stream, we can check the validity of its output as follows: if the fast algorithm recovers a vector $y^*$ which is $k$-sparse and has the same SIS sketch as the input, i.e., $(Ay^* = v\mod q)$, then by the correctness of Algorithm \ref{alg:k-rec}, $y^*$ equals the input. On the other hand, if $y^*$ is not $k$-sparse or its sketch does not match the SIS sketch $v$, then it must be that the input was not $k$-sparse. In both cases, Algorithm \ref{alg:fast} is $\poly(n)$ time and returns the correct result.
    Both recovery schemes from Algorithm \ref{alg:k-rec} and Theorem \ref{thm:fast} use  $\tilde{\mathcal{O}}(k)$ bits of space. Thus, Algorithm \ref{alg:fast} uses $\tilde{\mathcal{O}}(k)$ bits as well. Evaluating the output of the fast recovery scheme $(eval(\mathcal{F(\cdot)}))$ and  comparing the sketches takes $\poly(n)$ time, so the entire algorithm takes $\poly(n)$ time.
\end{proof}

\subsection{Applications of $k$-Sparse Recovery}

\subsubsection{Estimating the $\ell_0$ Norm}

Using our $k$-sparse recovery algorithm as a subroutine, we can construct an efficient $\ell_0$ estimation algorithm. This algorithm gives an $n^\eps$-approximation to the $\ell_0$ norm of a vector, whose entries are assumed to be bounded by $\poly(n)$. 

\begin{theorem}
     Under Assumption \ref{asm:strong}, for constant $\eps < 1$,  there exists a streaming algorithm robust against $o(n^{n^{1-\eps}})$ time-bounded white-box adversaries, which estimates the $L_0$ norm of a length-$n$ vector in the stream within a multiplicative factor of $n^{ \epsilon}$ using $\tilde{\mathcal{O}}(n^{1-\epsilon})$ bits of space and $\poly(n)$ time in the random oracle model.
 \end{theorem}

\input{algorithms/alg3}

\begin{proof}
    
    Algorithm \ref{alg:est} gives an $n^\eps$-approximation to the $\ell_0$ norm using Algorithm \ref{alg:k-rec} as a subroutine. Given a parameter $\eps > 0$, we set the parameter $k$ for $k$-sparse recovery to be $n^{1-\eps}$. The space used by Algorithm \ref{alg:est} is then  $\tilde{\mathcal{O}}(n^{1-\eps})$ bits. Also, both the recovery and the post-processing run in $\poly(n)$ time, so estimation can be done in $\poly(n)$ time. For correctness, if the vector is $n^{1-\eps}$-sparse, it can be recovered perfectly, and thus $\ell_0$($result$) is the exact value of the $\ell_0$ norm. Otherwise, for an input vector with more than $n^{1-\eps}$ non-zero entries, its $\ell_0$ norm lies in the range ($n^{1-\eps}$, $n$]. Hence, if we estimate its norm to be $n^{1-\eps}$, this gives an $n^\eps$-approximation. 
\end{proof}

\section{Matrix Recovery}

\subsection{Low-Rank Matrix Recovery}

In addition to recovering sparse vectors, we can recover low-rank matrices. We propose a white-box adversarially robust algorithm for the low-rank matrix recovery problem, which is efficient in terms of both time and space.

Similar to the $k$-sparse vector recovery problem, in order to achieve a fast update time while ensuring that the algorithm correctly detects inputs with rank larger than expected, we run two matrix recovery schemes in parallel. We will maintain one sketch based on a uniformly random matrix to distinguish if the input rank is too high to recover, and the other sketch will allow us to recover the input matrix if it is promised to be low rank. 

\begin{theorem}[\cite{doi:10.1137/070697835}]
\label{thm:iso}
    Let $\alpha = O(nk \log n)$ and let $A$ be a random matrix of dimension $\alpha \times n^2$, with entries sampled from an i.i.d. symmetric Bernoulli distribution:
    \[ A_{ij} = \begin{cases} 
      \sqrt{\frac{1}{\alpha}} & \text{with probability $\frac{1}{2}$} \\
      -\sqrt{\frac{1}{\alpha}} & \text{with probability $\frac{1}{2}$}
   \end{cases}\]
   Interpret $A$ as a linear map $\mathcal{A}: \mathbb{R}^{n \times n} \rightarrow \mathbb{R}^{\alpha}$ that computes $Ax$ for $x$ being the vectorization of an input $X \in \mathbb{R}^{n \times n}$. Then, given a rank-$r$ matrix $X_0 \in \mathbb{R}^{n \times n}$ and $b = \mathcal{A}(X_0)$ for $1 \leq r \leq min(k, n/2)$, with high probability $X_0$ is the unique low-rank solution to $\mathcal{A}(X) = b$ satisfying $rank(X) \leq r$. Moreover, $X_0$ can be recovered by solving a convex program: $argmin_X \lVert X\rVert_*$ subject to $\mathcal{A}(X)=b$.
\end{theorem}

We state our main theorem for matrix recovery:

\begin{theorem}
    Under Assumption \ref{asm:strong}, given an integer parameter $k$, there exists a streaming algorithm robust against $o(n^{nk})$ time-bounded white-box adversaries that either states that the input matrix has rank greater than $k$, or recovers the input matrix with rank at most $k$ using $\tilde{\mathcal{O}}(nk)$ bits of space and $\poly(n)$ time in the random oracle model.
\end{theorem}

\input{algorithms/alg4}

\begin{proof}
    Algorithm \ref{alg:m-rec} decides and recovers a matrix with rank no greater than $k$ using $\tilde{\mathcal{O}}(nk)$ bits of space.     
    For any input matrix $X \in \mathbb{Z}^{n \times n}_\beta$ with $\beta \in \poly(n)$, $q \geq n\cdot \beta$, and $rank(X) \leq k$, by the uniqueness of the low-rank solution given in Theorem \ref{thm:iso}, $X$ can be recovered by solving a convex program, and its product with the matrix $H$ matches the sketch $v$. On the other hand, when $rank(X) > k$, by Lemma \ref{lem:mat} under the SIS hardness assumption, there does not exist a low-rank matrix $Y$ distinct from $X$, for which $Hy = v = Hx\mod q$ with $x, y$ being the vectorization of $X, Y$, respectively. Therefore, in this case Algorithm \ref{alg:m-rec} outputs $None$, as desired.
    
    Both random matrices $H$ and $A$ used in Algorithm \ref{alg:m-rec} can be generated on the fly in the random oracle model. Therefore, the recovery algorithm only stores two sketch vectors of length $\tilde{\mathcal{O}}(nk)$ with entries bounded by $\poly(n)$, taking $\tilde{\mathcal{O}}(nk)$ bits in total. Solving the convex problem with the ellipsoid method and then comparing the solution with the sketch is $\poly(n)$ time, giving overall $\poly(n)$ time.
\end{proof}

\begin{rem} 
    We argue the optimality of our low-rank matrix recovery algorithm: with roughly $nk$ space, any white-box adversarially robust algorithm for low-rank matrix recovery has to assume that the adversary is $n^{nk}$-time bounded. Otherwise, the adversary has enough time to find a pair of inputs $X \neq X'$ that go to the same state and satisfy $rank(X' - X) > k$. Inserting $X$ then $-X$, or $X'$ then $-X$ into the stream, the algorithm cannot tell if the input is $0$ or $X'-X$. Hence, our $\tilde{\mathcal{O}}(nk)$-bit algorithm assuming an $o(2^{nk\log n}) = o(n^{nk})$ adversary is nearly optimal.
\end{rem}

\subsection{Applications of Low-Rank Matrix Recovery}

Our low-rank matrix recovery algorithm can be applied to a number of other problems on data streams. 

\subsubsection{Rank Decision Problem}

\begin{defn}[Rank Decision Problem]
Given an integer $k$, and an $n \times n$ matrix $A$, determine whether the rank of $A$ is larger than $k$.
\end{defn}

\begin{theorem}
    Under Assumption \ref{asm:strong}, given an integer parameter $k$, there exists a streaming algorithm robust against $o(n^{nk})$ time-bounded white-box adversaries that solves the rank decision problem using $\tilde{\mathcal{O}}(nk)$ bits of space and $\poly(n)$ time in the random oracle model.
\end{theorem}

\begin{proof}

    This problem is solved by running Algorithm \ref{alg:m-rec} with parameter $k$.
    This directly improves \cite{10.1145/3517804.3526228}.
\end{proof}

\subsubsection{Graph Matching}

\begin{defn}[Maximum Matching Problem]
Given an undirected graph $G = (V, E)$, the maximum matching problem is to find a maximum set of vertex disjoint edges in $G$. In a stream, we see insertions and deletions to edges. 
\end{defn}

\begin{theorem}
    Under Assumption \ref{asm:strong}, given an integer upper bound $k'$ on the size of a maximum matching in the graph, there is a streaming algorithm robust against $o(n^{2nk'})$ time-bounded white-box adversaries that finds a maximum matching in a graph using $\tilde{\mathcal{O}}(nk')$ bits of space and $\poly(n)$ time in the random oracle model.
\end{theorem}

\begin{proof}
    We can use the fact that the rank of the $n \times n$ Tutte matrix $A$ of the graph $G$, where $A_{i,j} = 0$ if there is no edge from $i$ to $j$, and $A_{i,j} = x_{i,j}$ and $A_{j,i} = -x_{i,j}$ for an indeterminate $x_{i,j}$ otherwise, equals twice the maximum matching size of $G$. Here the rank of $A$ is defined to be the maximum rank of $A$ over the reals over all assignments to its indeterminates. The main issue, unlike standard algorithms (see, e.g., Sections 4.2.1 and 4.2.2 of \cite{cheung13}), is that we cannot fill in the entries of $A$ randomly in a stream in the white-box model because the adversary can see our state and try to fool us. 
    Fortunately, there is a fix - in the stream we replace all $x_{i,j}$ deterministically with the number $1$. Call this deterministically filled in matrix $A'$, and note that the rank of $A'$ is at most the rank of $A$, and the latter is twice the maximum matching size. We then run our low-rank matrix recovery algorithm with parameter $k$ set to $2k'$. If we detect that the rank of $A'$ is greater than $2k'$, then the rank of $A$ is greater than $2k'$, and the maximum matching size is larger than $k'$ and we stop the algorithm and declare this. Otherwise, we have successfully recovered $A'$ and now that the stream is over, the locations of the $1$s are exactly the indeterminates in $A$, and so we have recovered $A$ and hence $G$ and thus can run any offline algorithm for computing a maximum matching of $G$.
\end{proof}

\subsection{Extension to Robust PCA and Tensors}

A recurring idea in our algorithms is to run two algorithms in parallel: (1) an algorithm to detect if the input is drawn from a small family of inputs, such as those which are sparse or low rank or both, and (2) a time-efficient deterministic algorithm which recovers the input if it is indeed drawn from such a family. The algorithm in (1) relies on the hardness of SIS while the algorithm in (2) is any time and space efficient deterministic, and thus white box adversarially robust, algorithm. For (1) we use an SIS matrix, and for (2), for robust PCA we use the algorithm of \cite{Tanner2020CompressedSO} while for tensors we use the algorithm of \cite{tensor}.

\subsubsection{Robust Principal Component Analysis}

The problem of Robust Principal Component Analysis is defined as follows:

\begin{defn}[Robust Principal Component Analysis]
Consider a data matrix 
    $M \in \mathbb{Z}_q^{n \times n}$ for $q \geq \poly(n)$, such that there exists a decomposition $M = L + S$, where $L \in \mathbb{Z}_q^{n \times n}$ satisfies $rank(L) \leq k$ and $S \in \mathbb{Z}_q^{n \times n}$ has at most $r$ non-zero entries. 
    The robust principal component analysis (RPCA) problem seeks to find the components $L$ and $S$.
\end{defn}

As in Section \ref{sec:sis}, we derive the following lemma based on the hardness of the SIS problem. 

\begin{lem}
    \label{lem:rpca}
    Under Assumption \ref{asm:strong}, given a uniformly random matrix $A \in \mathbb{Z}^{n \times m}_q$ for $q, m, \beta \in \poly(n)$ and $q \geq n\cdot \beta$, if a matrix $X \in \mathbb{Z}^{\sqrt{m} \times \sqrt{m}}_\beta$ is generated by an $o(2^{n})$-time adversary, then with probability $\geq 1-negl(n)$, there do not exist matrices $L, S \in \mathbb{Z}^{\sqrt{m} \times \sqrt{m}}_\beta$ with $rank(L) \leq k$ and\footnote{We use $nnz(S)$ to denote the number of non-zero entries of a matrix $S$.} $nnz(S) \leq r$, for which $X \neq L+S\mod q$ and $Ax = A(l+s)\mod q$, for $x, l, s$ being the vectorization of $X, L$ and $S$, respectively, and $r, k$ satisfying $k \in o(\frac{n-r\log n}{\sqrt{m}\log n})$.
\end{lem}

\begin{proof} 
Similar to the proof of Lemma \ref{lem:mat}, an adversary is able to try all pairs of matrices $L, S \in \mathbb{Z}^{\sqrt{m} \times \sqrt{m}}_\beta$ with $rank(L) \leq k$ and $nnz(S) \leq r$ in $\poly(n)^{r+k\sqrt{m}}$ time. As shown in the proof of Lemma \ref{lem:mat}, there are $\poly(n)^{\sqrt{m}k}$-many candidates for $L$. For the sparse matrix $S$, there are ${m \choose r} \in \poly(n)^r$ ways of positioning the non-zero entries, with their values chosen in $\poly(n)$. Therefore in total there are $\poly(n)^{r+k\sqrt{m}}$ pairs of candidate matrices.

When $k \in o(\frac{n-r\log n}{\sqrt{m}\log n})$, there exists an $o(2^{n})$-time adversary that is able to iterate through all candidate pairs. Thus, under Assumption \ref{asm:strong}, with overwhelming probability such an $L, S$ do not exist, otherwise, given $L$ and $S$, an adversary can solve the SIS problem by outputting $(x$-$l$-$s)\mod q$.
\end{proof}

As in our matrix recovery algorithm, we run a compressed sensing scheme for RPCA in parallel to achieve a fast recovery time. This fast scheme approximates a unique pair of low rank and sparse matrices from their sum,  assuming the sum is decomposable into a pair of such matrices. 

\begin{theorem}[\cite{Tanner2020CompressedSO}]
\label{thm:iso-rs}
    Let $\alpha = O((nk+r)\cdot \log n)$, and 
    let $A$ be a random matrix of dimension $\alpha \times n^2$, with entries sampled from an i.i.d. symmetric Bernoulli distribution:
    \[ A_{ij} = \begin{cases} 
      \sqrt{\frac{1}{\alpha}} & \text{with probability $\frac{1}{2}$} \\
      -\sqrt{\frac{1}{\alpha}} & \text{with probability $\frac{1}{2}$}
   \end{cases}\]
   Interpret $A$ as a linear map $\mathcal{A}: \mathbb{R}^{n \times n} \rightarrow \mathbb{R}^{\alpha}$ which computes $Ax$ for $x$ being the vectorization of an input $X \in \mathbb{R}^{n \times n}$. Then given $b = \mathcal{A}(L_0 + S_0)$, with high probability, $L_0, S_0$ is the unique solution to $\mathcal{A}(L+S) = b$ satisfying $rank(L_0) \leq k$ and $nnz(S) \leq r$. Moreover, $L_0, S_0$ can be recovered efficiently to a precision of $\lVert (L+S) -(L_0+S_0) \rVert_F \leq 42\eps$ by solving a semidefinite program: 
   \[argmin_{L, S}(\lVert L \rVert_* + \sqrt{2r/s}\cdot\lVert S \rVert_1)\] subject to 
   \[\lVert \mathcal{A} (L+S) -b \rVert_2 \leq \eps\]
   with the nuclear norm $\lVert \cdot \rVert_*$ of a matrix $M$ defined as the sum of its singular values $\lVert M \rVert_* = \sum_i\sigma_i(M)$; and the $1$-norm $\lVert \cdot \rVert_1$ defined as its maximum absolute column sum $\lVert M \rVert_1 = max_{0 \leq j \leq n} \sum_{i = 1}^n \lvert M_{ij}\rvert$.
\end{theorem}

Note that for an integer stream, we can set the error parameter $\eps \leq \frac{1}{\poly(n)}$ and then round the entries of the result to integers to guarantee exact recovery.

With that, we state our theorem for robust PCA:

\begin{theorem}
\label{thm:rpca-1}
    Under Assumption \ref{asm:strong}, given parameters $r, k > 0$, there exists a streaming algorithm robust against $o(n^{nk+r})$ time-bounded white-box adversaries that determines if an $n \times n$ input matrix can be decomposed into the sum of a matrix with rank at most $k$ and a matrix with at most $r$ non-zero entries, and if so, finds the decomposition using $\tilde{\mathcal{O}}(nk+r)$ bits of space and $\poly(n)$ time in the random oracle model.
\end{theorem}

\input{algorithms/alg5}

\begin{proof}
    Algorithm \ref{alg:rpca} determines if an $n \times n$ input matrix can be decomposed into a low rank matrix plus a sparse matrix.
    For any input matrix $X_0 = L_0+S_0 \in \mathbb{Z}^{n \times n}_\beta$ with $rank(L_0) \leq k$ and $nnz(S_0) \leq r$, by the uniqueness of the solution pair given in Theorem \ref{thm:iso-rs}, $L_0, S_0$ can be recovered by solving a semidefinite program, and the product of their sum with the matrix $H$ matches the sketch $v$. On the other hand, when the input $X$ cannot be decomposed into low rank and sparse components, by Lemma \ref{lem:rpca} under the SIS hardness assumption, there does not exist a pair of low rank and sparse matrices $L', S'$ such that $X \neq L' + S'$ and $H(l'+s') = v = Hx\mod q$, for $l', s', x$ being the vectorization of $L', S', X$, respectively. Therefore, in this case Algorithm \ref{alg:rpca} outputs $None$, as desired.
    
    Both random matrices $H$ and $A$ used in Algorithm \ref{alg:rpca} can be generated on the fly in the random oracle model. Therefore, the recovery algorithm only stores two sketch vectors of length $\tilde{\mathcal{O}}(nk+r)$ with entries bounded by $\poly(n)$, taking $\tilde{\mathcal{O}}(nk+r)$ bits in total. Also, solving the semidefinite program and then comparing the solution with the sketch takes $\poly(n)$ time, giving overall $\poly(n)$ time.
\end{proof}

\subsubsection{Tensor Recovery}
\label{sup:tensor}

Similar to our vector and matrix recovery algorithms, we propose an algorithm that recovers tensors with low CANDECOMP/PARAFAC (CP) rank.

\vspace{0.2em}
{\bf Notation:}
    Let $\otimes$ denote the outer product of two vectors. Then one can build a rank-$1$ tensor in $\mathbb{Z}^{n_1 \times n_2 \times \cdots \times n_d}$ by taking the outer product $x_1 \otimes x_2 \otimes \cdots \otimes x_d$ where $x_i \in \mathbb{Z}^{n_i}$. 

\begin{defn}[CP-rank]
    For a tensor $X \in \mathbb{Z}^{n_1 \times \cdots \times n_d}_q$, consider it to be the sum of $r$ rank-1 tensors:
    $X = \sum_{i=1}^r(x_{i1} \otimes x_{i2} \otimes \cdots  
    \otimes x_{id})$ where $x_{ij} \in \mathbb{Z}^{n_j}_q$.
    The smallest number of rank-$1$
    tensors that can be used to express a tensor $X$ is then defined to be the rank of the tensor. 
\end{defn}

As in Section \ref{sec:sis}, we derive the following lemma based on the hardness of the SIS problem. 

\begin{lem}
    \label{lem:tensor}
    Under Assumption \ref{asm:strong}, given a uniformly random matrix $A \in \mathbb{Z}^{n \times m}_q$ for $q, m, \beta \in \poly(n)$ and $q \geq n\cdot \beta$, if a tensor $X \in \mathbb{Z}^{n_1 \times \cdots \times n_d}_\beta$ is generated by an $o(2^{n})$-time adversary, where $\prod n_i = m$, then with probability $\geq 1-negl(n)$, there does not exist a tensor $Y \in \mathbb{Z}^{n_1 \times \cdots \times n_d}_\beta$ with $rank(Y) \leq k$, such that $X \neq Y\mod q$ and $Ax = Ay\mod q$, for $x, y$ being the vectorization of $X$ and $Y$, respectively, $\prod n_i = m$, and $k \in o(\frac{n}{(n_1 + \cdots + n_d)\log n})$.
\end{lem}

\begin{proof} 
Similar to the proof of Lemma \ref{lem:mat}, an adversary is able to try all low rank tensors $Y \in \mathbb{Z}^{n_1 \times \cdots \times n_d}_\beta$ with $rank(Y) \leq k$ in $\poly(n)^{k(n_1 + \cdots + n_d)}$ time. For each $x_{ij}$, there are $\poly(n)^{n_j}$ choices of its value. So we have $\poly(n)^{n_1+ \cdots +n_d}$ many possible rank-$1$ tensors. Choosing $k$-many of them to generate a rank-$k$ tensor, that is $\poly(n)^{k(n_1 + \cdots + n_d)}$ candidates in total.

When $k \in o(\frac{n}{(n_1 + \cdots + n_d)\log n})$, there exists an $o(2^{n})$-time adversary that is able to iterate through all candidate pairs. Thus, under Assumption \ref{asm:strong}, with overwhelming probability such a $Y$ does not exist; otherwise, given $L$ and $S$, an adversary can solve the SIS problem by outputting $(x$-$y)\mod q$.
\end{proof}

As we did for vector and matrix recovery problems, we can run a fast low rank tensor estimation scheme in parallel in our tensor recovery algorithm. 

\begin{theorem}[\cite{tensor}]
\label{thm:iso-tensor}
    Let $\mathcal{A}:\mathbb{R}^{n_1 \times \cdots \times n_d} \rightarrow \mathbb{R}^{k(n_1 + \cdots n_d)logn}$ be a Gaussian measurement operator whose entries are properly normalized, i.i.d. Gaussian random variables. Then given a measurement $b = \mathcal{A}(X)$ for a tensor $X \in \mathbb{R}^{n1 \times \cdots \times n_d}$, there exists an algorithm that gives an estimate $X_0 \in \mathbb{R}^{n_1 \times \cdots \times n_d}$ with $\lVert X_0 - X \rVert_F \leq \frac{1}{\poly(n)}$ with high probability using $\poly(n)$ time, for $n = \prod^d_1n_i$.
\end{theorem}

\begin{rem}
\label{rem:tensor}
    For our purposes, we round the Gaussian random variables to additive integer multiples of $\frac{1}{poly(n)}$. This rounding changes the norm of the measurement by at most additive $\frac{1}{poly(n)}$, and therefore asymptotically does not change the result in Theorem \ref{thm:iso-tensor}. The discretized random variables can then be constructed to the desired precision using uniformly random bits \cite{10.1145/2710016} generated by a random oracle.

    Then running the algorithm from Theorem \ref{thm:iso-tensor} in a stream, we only have to maintain and update a measurement vector of length $\tilde{\mathcal{O}}({k(n_1 + \cdots n_d)})$ with entries bounded in $\poly(n)$, giving an overall $\tilde{\mathcal{O}}(k(n_1 + \cdots +n_d))$ bits of space usage.
\end{rem}

For an integer stream, we can round the entries of the estimation result to integers to guarantee exact recovery. We formulate the low rank tensor recovery algorithm as follows.

\begin{theorem}
\label{thm:tensor}
    For an $n_1 \times \cdots \times n_d$ input tensor $X$, under Assumption \ref{asm:strong}, given a parameter $k$ with $k \in \Theta(\frac{n^c}{(n_1+\cdots+n_d)\log n})$ for $n = \prod^d_1n_i$ and a constant $c > 0$, there exists a streaming algorithm robust against $o(n^{k(n_1 + \cdots + n_d)})$ time-bounded white-box adversaries that determines if the input tensor has CP rank at most $k$ and if so, recovers the tensor using $\tilde{\mathcal{O}}(k(n_1 + \cdots n_d))$ bits of space and $\poly(n)$ time in the random oracle model.
\end{theorem}

\input{algorithms/alg6}

\begin{proof}
    Algorithm \ref{alg:tensor} determines if an $n_1 \times n_2 \times \cdots \times n_d$ input tensor has rank at most $k$ and if so, recovers the input tensor. 
    
    For any input tensor $X \in \mathbb{Z}^{n_1 \times \cdots \times n_d}_\beta$ with $rank(X) \leq k$, by Theorem \ref{thm:iso-tensor}, $eval(\mathcal{F}(\cdot))$ correctly reconstructs it. Also, the product of it with the matrix $H$ matches the sketch $v$. On the other hand, when the input $rank(X) > k$, by Lemma \ref{lem:tensor} under the SIS hardness assumption, there does not exist a tensor $Y$ with $rank(Y) \leq k$ such that $X \neq Y$ and $Hy = v = Hx\mod q$, for $x, y$ being the vectorization of $X, Y$, respectively. Therefore, in this case Algorithm \ref{alg:tensor} outputs $None$, as desired.
    
    The random matrix $H$ used in Algorithm \ref{alg:tensor} can be generated on the fly in the random oracle model. Therefore, the recovery algorithm only stores a sketch vector of length $\tilde{\mathcal{O}}(k(n_1 + \cdots +n_d))$ with entries bounded by $\poly(n)$. Also, as stated in Remark \ref{rem:tensor}, the fast recovery scheme $\mathcal{F}(\cdot)$ takes $\tilde{\mathcal{O}}(k(n_1 + \cdots +n_d))$ bits of space, so the total space usage is $\tilde{\mathcal{O}}(k(n_1 + \cdots +n_d))$. Both the evaluation of the fast recovery scheme $eval(\mathcal{F}(\cdot))$ and the comparison of vectors take $\poly(n)$ time, giving overall $\poly(n)$ time.
\end{proof}

\section{Conclusion}
We give robust streaming algorithms against computationally bounded white-box adversaries under cryptographic assumptions. We design efficient recovery algorithms for vectors, matrices, and tensors which can detect if the input is not sparse or low rank. We use these to improve upon and solve new problems in linear algebra and optimization, such as detecting and finding a maximum matching if it is small. 
It would be interesting to explore schemes that can recover vectors that are only approximately $k$-sparse or matrices that are only approximately rank-$k$. We make progress on the latter by considering robust PCA, but there is much more to be done. Also, although our algorithm improves the space-accuracy trade-off for $\ell_0$-norm estimation, it is unclear if it is optimal, and it would be good to generalize to $\ell_p$ norms for $p > 0$, as well as other statistics of a vector.

\nocite{langley00}

\section*{Acknowledgements}
We thank Aayush Jain for a helpful discussion on pseudorandom functions and the random oracle model. D. Woodruff would like to thank support from the National Institute of Health (NIH) grant 5R01 HG 10798-2 and a Simons Investigator Award. 

\def\shortbib{0}
\bibliographystyle{alpha}
\bibliography{ref}

\end{document}

%% file: preamble.tex
\pdfoutput=1
\usepackage[protrusion=true,expansion=true]{microtype}
\usepackage[dvipsnames]{xcolor}
\usepackage{amsmath,amssymb,amsfonts,amsthm}
\usepackage{mathtools}
\usepackage{subcaption}
\usepackage{float}
\usepackage{graphicx}
\usepackage[backref=page]{hyperref}
\usepackage{color}
\usepackage{wrapfig}
\usepackage{tikz}
\usetikzlibrary{decorations.pathreplacing}
\usepackage{setspace}
\usepackage[ruled, linesnumbered]{algorithm2e}
\usepackage[framemethod=tikz]{mdframed}
\usepackage{xspace}
\usepackage{pgfplots}
\usepackage{framed}
\usepackage{subcaption}
\usepackage{thmtools}
\usepackage{thm-restate}
\usepackage{comment}
\usepackage{bm}
\usepackage{enumitem}
\usepackage{booktabs}

\pgfplotsset{compat=1.5}

\newtheorem{theorem}{Theorem}[section]

\newtheorem{lem}[theorem]{Lemma}
\newtheorem{defn}[theorem]{Definition}
\newtheorem{asm}[theorem]{Assumption}
\newtheorem{rem}[theorem]{Remark}

\newenvironment{proofof}[1]{\begin{trivlist} \item {\bf Proof
#1:~~}}
  {\qed\end{trivlist}}

\newcommand{\namedref}[2]{\hyperref[#2]{#1~\ref*{#2}}}

\newcommand{\eps}{\varepsilon}

\DeclareMathOperator*{\poly}{poly}

\newcommand{\ignore}[1]{}

\newif\ifnotes\notestrue 
\ifnotes
\newcommand{\samson}[1]{\textcolor{purple}{{\bf (Samson:} {#1}{\bf ) }} \marginpar{\tiny\bf
             \begin{minipage}[t]{0.5in}
               \raggedright S:
            \end{minipage}}}            							
\else
\newcommand{\samson}[1]{}
\fi

\hypersetup{
     colorlinks   = true,
     citecolor    = blue,
		 linkcolor		= red
}

\makeatletter
\renewcommand*{\@fnsymbol}[1]{\textcolor{mahogany}{\ensuremath{\ifcase#1\or *\or \dagger\or \ddagger\or
 \mathsection\or \triangledown\or \mathparagraph\or \|\or **\or \dagger\dagger
   \or \ddagger\ddagger \else\@ctrerr\fi}}}
\makeatother

\providecommand{\email}[1]{\href{mailto:#1}{\nolinkurl{#1}\xspace}}

\definecolor{mahogany}{rgb}{0.75, 0.25, 0.0}
\definecolor{bleudefrance}{rgb}{0.19, 0.55, 0.91}
\hypersetup{
     colorlinks   = true,
     citecolor    = bleudefrance,
  	 linkcolor	  = bleudefrance
}

\DeclarePairedDelimiterX{\inp}[2]{\langle}{\rangle}{#1, #2}

\DeclarePairedDelimiterX{\infdivx}[2]{(}{)}{%
  #1\;\delimsize\|\;#2%
}

\numberwithin{equation}{section}

\usepackage{cleveref}

%% file: table.tex
\tabcolsep=0.11cm
\begin{table*}[t]
    \caption{A summary of the bit complexities of our algorithms, as compared to the best known upper bounds for these problems in the white-box adversarial streaming model. Dash means that we provide the first algorithm for the problem in the white-box stream model. For $k$-sparse recovery, we require $k \geq n^c$ for an arbitrarily small constant $c>0$.}
    \label{tab:sum}
    \vskip 0.15in
    \begin{center}
    \begin{small}
    \begin{sc}
    \scalebox{0.92}{
    \begin{tabular}{lcccr}
    \toprule
    Problem & Previous Space & Our Space & Note \\
    \midrule
    K-sparse recovery & -- & $\tilde{\mathcal{O}}(k)$ & detects dense input\\
    L$_0$-norm estimation & $\tilde{\mathcal{O}}(n^{1-\eps+c\eps})$ & $\tilde{\mathcal{O}}(n^{1-\eps})$ & achieves $n^{\epsilon}$-approximation \\
    Low-rank matrix recovery & -- & $\tilde{\mathcal{O}}(nk)$ & detects high-rank input\\
    Low-rank tensor recovery & -- & $\tilde{\mathcal{O}}(k(n_1+\cdots+n_d))$ & detects high CP-rank input\\
    Robust PCA & -- & $\tilde{\mathcal{O}}(nk+r)$ & Detects Not Sparse + Low Rank\\
    Rank-decision & $\tilde{\mathcal{O}}(nk^2)$ & $\tilde{\mathcal{O}}(nk)$ & detects high-rank input\\
    Maximum matching & -- & $\tilde{\mathcal{O}}(nk)$ & Detects Large matching size\\
    \bottomrule
    \end{tabular}}
    \end{sc}
    \end{small}
    \end{center}
    \vskip -0.1in
\end{table*}

%% file: algorithms/alg1.tex
\begin{algorithm}
\caption{Recover-Vector($n$, $m$, $k$)}
\label{alg:k-rec}

\KwData{$m$ integer updates $u_t$ to a length-$n$ vector with entries bounded in $\beta \in \poly(n)$;
\newline
A modulus $q \in \poly(n)$ with $q \gg \beta$.}

Let $f(k)$ be a function in $\omega(k)$ and $\Tilde{\mathcal{O}}(k)$. Initialize a uniformly random matrix $A \in \mathbb{Z}^{(f(k)\cdot \log n) \times n}_{q}$ for $q \in \poly(n)$ and a zero vector $v$ of length $k\cdot \log n$\;

\ForEach{update $u_t$ with $t\in [m]$}{
     Update $v$ by adding $u_t\cdot A_i$ to it, where $A_i$ is the $i^{th}$ column of $A$, and where the stream update changes the $i^{th}$ coordinate by an additive amount $u_t \in \mathbb{Z}_q$\;}

\ForEach{$k$-sparse vector $y$ with entries $\in [-\beta, \beta]$}{
    \If{$Ay = v\mod q$}{
    \Return $y$\;
    }}

\Return $None$\;

\end{algorithm}

%% file: algorithms/alg2.tex
\begin{algorithm}
\caption{Fast-Recover($n$, $m$, $k$)}
\label{alg:fast}

\KwData{$m$ integer updates $u_t$ to a length-$n$ vector with entries bounded in $\beta \in \poly(n)$;
\newline
A modulus $q \in \poly(n)$ with $q \gg \beta$.}

Initiate an instance of the fast $k$-sparse recovery scheme $\mathcal{F}(\cdot)$ from Theorem \ref{thm:fast}\;

Let $f(k)$ be a function in $\omega(k)$ and $\Tilde{\mathcal{O}}(k)$. Initialize a uniformly random matrix $A \in \mathbb{Z}^{(f(k)\cdot \log n) \times n}_{q}$ for $q \in \poly(n)$ and a zero vector $v$ of length $k\cdot \log n$\;

\ForEach{update $u_t$ with $t\in [m]$}{
    Feed the update to the initiated instance $\mathcal{F}(\cdot)$\;
    
     Update $v$ by adding $u_t\cdot A_i$ to it, where $A_i$ is the $i^{th}$ column of $A$, and where the stream update changes the $i^{th}$ coordinate by an additive amount $u_t \in \mathbb{Z}_q$\;}

$y^*$ $\gets$ $eval(\mathcal{F}(\cdot))$\;

\If{$y^*$ is $k$-sparse \textbf{andalso} $\lVert y^*\rVert_\infty \leq \beta$ \textbf{andalso} $Ay^* = v\mod q$}{
    \Return $y^*$\;
}

\Return $None$\;

\end{algorithm}

%% file: algorithms/alg3.tex
\begin{algorithm}
\caption{Estimate-L0($n$, $m$, $\eps$)}
\label{alg:est}

\KwData{$m$ integer updates $u_t$ to a length-n vector.}

$result$ $\gets$ Fast-Recover($n, m, n^{1-\eps}$)\;

\If{$result = None$}{
\Return $n^{1-\eps}$\;
}
\Return $\ell_0$($result$)\;

\end{algorithm}

%% file: algorithms/alg4.tex
\begin{algorithm}
\caption{Recover-Matrix($n$, $m$, $k$)}
\label{alg:m-rec}

\KwData{$m$ integer updates $u_t$ to an $n \times n$ matrix with entries bounded in $\beta \in \poly(n)$;
\newline
A modulus $q \in \poly(n)$ with $q \gg \beta$.}

 Let $f(k)$ be a function in $\omega(k)$ and $\Tilde{\mathcal{O}}(k)$. Initialize a uniformly random matrix $H \in \mathbb{Z}^{f(k)\cdot n\log n \times n^2}_{q}$ for $q \in \poly(n)$, a matrix $A:\alpha \times n^2$ as specified in Theorem \ref{thm:iso}, and zero vectors $v, w$ of length $f(k)\cdot n\log n$\;

\ForEach{update $u_t$ with $t\in [m]$}{
    Update $v$ by adding $u_t\cdot H_i$ to it, and update $w$ by  adding $u_t \cdot A_i$ to it, where $i$ corresponds to the vectorized index of the update, and where $H_i, A_i$ are the $i^{th}$ columns of $H, A$, respectively\;}

$X_0 \gets argmin_X \lVert X\rVert_*$ subject to $A\cdot vectorize(X)=w$\;

\If{$rank(X_0) \leq k$ \textbf{andalso} $X_0 \in \mathbb{Z}^{ n \times n}_{\beta}$ \textbf{andalso} $H\cdot vectorize(X_0)=v\mod q$}{
    \Return $X_0$\;
}

\Return $None$\;

\end{algorithm}

%% file: algorithms/alg5.tex
\begin{algorithm}
\caption{RPCA($n$, $m$, $k$, $r$)}
\label{alg:rpca}

\KwData{$m$ integer updates $u_t$ to an $n \times n$ matrix with entries bounded in $\beta \in \poly(n)$;
\newline
A modulus $q \in \poly(n)$ with $q \gg \beta$.}

 Let $f(k)$ be a function in $\omega(k)$ and $\Tilde{\mathcal{O}}(k)$. Initialize a uniformly random matrix $H \in \mathbb{Z}^{(f(k)\cdot n + r)\log n \times n^2}_{q}$ for $q \in \poly(n)$, a fast recovery matrix $A:(nk+r)\log n \times n^2$ as specified in Theorem \ref{thm:iso-rs}, and zero vectors $v, w$ of length $(f(k)\cdot n + r)\log n$\;

\ForEach{update $u_t$ with $t\in [m]$}{
    Update $v$ by adding $u_t\cdot H_i$ to it, and update $w$ by  adding $u_t \cdot A_i$ to it, where $i$ corresponds to the vectorized index of the update, and where $H_i, A_i$ are the $i^{th}$ columns of $H, A$, respectively\;}

$L_0, S_0 \gets argmin_{L, S}(\lVert L \rVert_* + \sqrt{2r/s}\cdot\lVert S \rVert_1)$ subject to $\lVert \mathcal{A} (L+S) -b \rVert_2 \leq \frac{1}{\poly(n)}$\;

\If{$rank(L_0) \leq k$ \textbf{andalso} $nnz(S_0) \leq r$ \textbf{andalso} $L_0, S_0 \in \mathbb{Z}^{ n \times n}_{\beta}$ \textbf{andalso} $H\cdot vectorize(L_0 + S_0)=v\mod q$}{
    \Return $L_0, S_0$\;
}

\Return $None$\;

\end{algorithm}

%% file: algorithms/alg6.tex
\begin{algorithm}
\caption{Recover-tensor($n_1, \cdots n_d$, $m$, $k$)}
\label{alg:tensor}

\KwData{$m$ integer updates $u_t$ to an $n_1 \times \cdots \times n_d$ tensor with entries bounded in $\beta \in \poly(n)$;
\newline
A modulus $q \in \poly(n)$ with $q \gg \beta$.}

Initiate an instance of the fast low rank tensor recovery scheme $\mathcal{F}(\cdot)$ from Theorem \ref{thm:iso-rs}\;

Let $n = \prod^d_1 n_i$. Let $f(k)$ be a function in $\omega(k)$ and $\Tilde{\mathcal{O}}(k)$. Initialize a uniformly random matrix $H \in \mathbb{Z}^{f(k)(n_1 + \cdots +n_d)\log n \times n}_{q}$ for $q \in \poly(n)$, and a zero vector $v$ of length $f(k)(n_1 + \cdots  +n_d)\log n$\;

\ForEach{update $u_t$ with $t\in [m]$}{
    Feed the update to the initiated instance $\mathcal{F}(\cdot)$\;
    
     Update $v$ by adding $u_t\cdot A_i$ to it, where $H_i$ is the $i^{th}$ column of $H$, and where the stream update changes the $i^{th}$ coordinate by an additive amount $u_t \in \mathbb{Z}_q$\;}

$X^*$ $\gets$ $eval(\mathcal{F}(\cdot))$

\If{$rank(X^*) \leq k$ \textbf{andalso} $X^* \in \mathbb{Z}^{n_1 \times \cdots \times n_d}_{\beta}$ \textbf{andalso} $H\cdot vectorize(X^*)=v\mod q$}{
    \Return $X^*$\;
}

\Return $None$\;

\end{algorithm}